\documentclass[letterpaper, 10 pt, conference]{ieeetran}  % 
\IEEEoverridecommandlockouts                              % 

\usepackage{graphicx}      % include this line if your document contains figures
\usepackage{amsmath}
\usepackage{mathrsfs}
\usepackage{amsfonts}
\usepackage{amsthm}
\usepackage{amssymb}
\usepackage{color}

\usepackage{enumitem}
\usepackage{hyperref}

\def\rea{\mathbb{R}}
\newtheorem{proposition}{Proposition}
\newtheorem{remark}{Remark}
\newtheorem{theorem}{Theorem}
\newtheorem{assumption}{Assumption}
\newtheorem{corollary}{Corollary}

\begin{document}
\title{Tuning Rules for a Class of Passivity-Based Controllers for Mechanical Systems}
\author{Carmen Chan-Zheng, Pablo Borja, and Jacquelien M.A. Scherpen
	\thanks{Manuscript received on August 31, 2020; revised on November 10, 2020; accepted on December 3, 2020. The work of Carmen Chan-Zheng is supported by the University of Costa Rica. The authors are with the Jan C. Willems Center for Systems and 	Control, Engineering and Technology Institute Groningen, Faculty of Science and Engineering, University of Groningen,  9747 AG Groningen, The Netherlands (email: \texttt{c.chan.zheng@rug.nl,l.p.borja.rosales@rug.nl, j.m.a.scherpen@rug.nl}).}}

\maketitle

	\begin{abstract}
	In this paper, we propose several rules to tune the gains for a class of passivity-based controllers for nonlinear mechanical systems. Such tuning rules prescribe a desired \textit{local} transient response behavior to the closed-loop system. To establish the tuning rules, we implement a PID passivity-based controller. Then, we linearize the closed-loop system, and we transform the matrix of the resulting system into a class of \textit{saddle point matrices} to analyze the influence of the control gains, in terms of the oscillations and the rise time, on the transient response of the closed-loop system.  Hence, the resulting controllers stabilize the plant and simultaneously address the performance of the closed-loop system. Moreover, our analysis provides a clear insight into how the kinetic energy, the potential energy, and the damping of the mechanical system are related to its transient response, endowing in this way the tuning rules with a physical interpretation. Additionally, we corroborate the analytical results through the practical implementation of a controller that stabilizes a two degrees-of-freedom (DoF) planar manipulator, where the control gains are tuned following the proposed rules.
\end{abstract}

%===============================================================================
\section{Introduction}
New technological trends have created new control challenges in which current linear techniques are not adequate as the nonlinearities phenomena are no longer negligible. Nonetheless, in contrast to the linear methods, the development of a general framework to control nonlinear systems is still an open question. Thus, the current nonlinear control techniques are available only for special classes of systems. Furthermore, the vast majority of the nonlinear control methods only focus on the stability of the closed-loop system without providing any insight into how to tune the control gains, and consequently, disregarding some indicators of performance of the closed-loop system. Nevertheless, in several cases, it is essential to ensure a prescribed performance to solve a task at hand, e.g., applications involving physical systems that require high precision such as those found in aerospace, medical, semiconductor manufacturing, among other industries.

Passivity-based control (PBC) methods offer a constructive approach to control complex physical systems where the exchange of energy between the plant and the environment plays a central role \cite{ortega2013passivity}. Additionally, the gains of such controllers may be associated with the physical quantities of the closed-loop system, i.e., energy and damping. While there exist several references (e.g., \cite{van2000l2}, \cite{zhang2017pid}, \cite{borja2020}) where the authors implement PBC techniques to ensure that the closed-loop system exhibits a desired performance (e.g., $\mathcal{L}_2$ stability), the results related to gain tuning for removing oscillations and, or improving the rise time are scarce for PBC-methodologies.
Some works that we find in this line of research are \cite{jeltsema2004tuning} where the authors propose a methodology for tuning the damping gain in switched-mode power converters, \cite{dirksz2013tuning} where the authors improve the transient response of a class of mechanical system by modifying the initial conditions of a dynamic controller, and \cite{kotyczka2013local} where the author proposes a methodology--based on the linearization of the system--to tune an interconnection and damping assignment (IDA) PBC.

On the linear counterpart, several control design approaches exist such that the resulting closed-loop system exhibits a desired performance. In particular, for PID controllers, we find the Ziegler-Nichols methods for single-input single-output (SISO) systems \cite{ziegler1993optimum}. While for multiple-input multiple-output (MIMO) systems, we have diverse techniques as the ones reported in \cite{skogestad2003simple}, \cite{boyd2016mimo}, and \cite{ruiz2006robust}. However, some disadvantages of these methods include using heuristic approaches to derive the rules, employing a first or second-order time-delay model to approximate the real plant, or solving complex optimization problems that involve linear matrix inequalities.

Due to the simple structure of PID controllers and the suitability of PBC techniques to stabilize physical systems, recently, several authors have paid particular attention to the so-called PID-PBC approach (see \cite{zhang2017pid}, \cite{borja2020}, \cite{romero2017global}, \cite{jayawardhana2007passivity}). Some remarkable properties of PID-PBCs are that i) in contrast to other PBC techniques, the PID-PBC method does not require the solution of partial differential equations (PDEs), and ii) the tune of the control gains is trivialized for \textit{stability} purposes.	Nonetheless, to the best of the authors' knowledge, there is no available literature that provides guidelines on how to tune the gain of PID-PBCs to prescribe a desire performance in terms of oscillations, damping ratio, or rise time. 

Motivated by the lack of tuning guidelines for the PID-PBC technique, the main contribution of this manuscript is a set of tuning rules for the class of controllers described in \cite{zhang2017pid} and \cite{borja2020}, that prescribes a desired local behavior of the closed-loop system in terms of oscillations, damping ratio, or rise time. The rules are obtained by inspecting the linearized closed-loop system, where we analyze it via \textit{the saddle point matrix theory}. Hence, the resulting {PID-PBCs} are suitable to stabilize MIMO systems while ensuring a desired local transient response behavior of the closed-loop system. In contrast to the results reported in \cite{kotyczka2013local}, the stability analysis of the closed-loop system does not rely on the linearization of the system nor the solution of PDEs.

The remainder of this paper is structured as follows: in Section \ref{preliminaries}, we provide the preliminaries and the problem formulation. In Section \ref{tuningrules}, we present the main results of this paper.  In Section \ref{example}, we apply our tuning rules to control a 2DoF planar manipulator. We finalize the paper with some concluding remarks and future work in Section \ref{conclusion}.

\textbf{Notation}: We denote the $n\times n$ identity matrix as $I_n$ and the $n\times m$ matrix of zeros as $0_{n\times m}$. For a given smooth function $f:\mathbb{R}^n\to \mathbb{R}$, we define the differential operator $\nabla_x f:=(\frac{\partial f}{\partial x})^\top$ and $\nabla^2_x f:=\frac{\partial^2 f}{\partial x^2}$. For a smooth mapping $f:\mathbb{R}^n\to\mathbb{R}^m$, we define the $ij-$element of its $n\times m$ Jacobian matrix as $(\nabla_x f)_{ij}:=\frac{\partial f_j}{\partial x_i}$. When clear from the context the subindex in $\nabla$ is omitted. Given a distinguished element $x_\star \in \rea^n$, we define the constant matrix $B_\star:=B(x_\star)\in \rea^{n\times m}$.  
For a given matrix $A\in\mathbb{R}^{n\times n}$ and a vector $x\in\rea^{n}$, we say that $A$ is \textit{positive definite (semi-definite)}, denoted as $A>0$  ($A\geq0$), if $A=A^{\top}$ and $x^{\top}Ax>0$ ($x^{\top}Ax\geq0$) for all  $x\in \rea^{n}-\{0_{n} \} (\rea^{n}$). For a positive (semi-)definite matrix $A$,  we define the weighted Euclidean norm as $\lVert x \rVert_{A}:=\sqrt{x^{\top}Ax}$. For $A=A^\top$, we denote by $\lambda_{\min}(A)$ and $\lambda_{\max}(A)$ as the minimum eigenvalue and the maximum eigenvalue of $A$, respectively. Consider $y\in \mathbb{C}^n$, we denote by $y^*$ as the conjugate transpose vector of $y$.
%===============================================================================
\section{Preliminaries and problem setting} \label{preliminaries}
In this section, we summarize some properties of a class of \textit{saddle point matrices}, which are the cornerstone in the development of the tuning rules presented in Section \ref{tuningrules}. Although not mentioned explicitly in the pioneering work of Brayton and Moser \cite{brayton1964theory}, the authors use these properties to verify the behavior of the transient response.
Then, we provide the port-Hamiltonian (pH) representation of the class of mechanical systems for which our tuning rules are suitable. Finally, we describe some details of the PID-PBC implemented in this work. 

\subsection{Some properties of a class of saddle point matrices}
Consider the following linear system
\begin{equation}\label{sys}
	\begin{bmatrix} \dot{\tilde{x}}\\ \dot{\tilde{y}}\end{bmatrix}=-\Phi\begin{bmatrix}\tilde{x}\\ \tilde{y}\end{bmatrix}, \quad \Phi:=\begin{bmatrix}X&Z^\top\\-Z&Y\end{bmatrix} 
\end{equation}
where  $X \in \mathbb{R}^{n\times n}$ is positive-definite, $Y \in \mathbb{R}^{m\times m}$ is positive semi-definite, $Z \in \mathbb{R}^{m\times n}$ has full row rank, $\tilde{x} \in \mathbb{R}^n$, and $\tilde{y} \in \mathbb{R}^m$, where  $m\leq n$. The structure of $\Phi$ corresponds to \textit{a class of saddle point matrices}. Hence, the spectrum of  $\Phi$ can be analyzed via inspection of the spectrum of $X$, $Y$, and $Z$. In particular, we are interested in the results below in Theorem \ref{t1} and Corollary \ref{t2}.
\begin{theorem}[\cite{benzi2006eigenvalues}]\label{t1}
	Let $Y=0_{m \times m}$, $v\in\mathbb{C}^n$, $w\in\mathbb{C}^m$. Denote with $\lambda_{\Phi}$ and $(v^\top, w^\top)^\top$ an eigenvalue and an eigenvector of $\Phi$, respectively\footnote{The vector $w$ is not used explicitly in this manuscript but is essential to prove this theorem. See \cite{benzi2006eigenvalues} for further details.}. Then, $\lambda_{\Phi}$ is real if and only if:
	\begin{align}\label{condt1}
		\Bigg(\frac{v^* X v}{v^* v}\Bigg)^2\geq 4 \frac{v^*(Z^\top Z)v}{v^* v}.
	\end{align}
\end{theorem}

Furthermore, if Corollary $2.6$ of \cite{benzi2006eigenvalues} is satisfied, then, the entire spectrum of $\Phi$ is real.
\begin{corollary}[\cite{shen2010eigenvalue}]\label{t2}
	Let $Y=0_{m \times m}$. Denote with $\lambda_{\Phi} \in \mathbb{C}$ any eigenvalue of $\Phi$. Then, the following statements are true:
	\begin{enumerate}[label=\roman*),wide, labelwidth=!, labelindent=0pt]
		\item If $\Im(\lambda_{\Phi})\neq0$, then $				\frac{1}{2}\lambda_{\min}(X)\leq\Re(\lambda_{\Phi})\leq\frac{1}{2}\lambda_{\max}(X)
		$
		\item If $\Im(\lambda_{\Phi})=0$, then 
		\begin{equation*}\label{t2eq2}
			\min\{\lambda_{\min}(X),\lambda_{\min}(ZX^{-1}Z^\top)\}\leq\lambda_{\Phi}\leq\lambda_{\max}(X).
		\end{equation*}
	\end{enumerate}	
\end{corollary}

\subsection{PID-PBC for mechanical systems} 
Throughout this work, we consider mechanical systems that admit a pH representation of the form
\begin{equation}\label{phsys}
	\begin{split}
		&\begin{bmatrix}
			\dot{q}\\ \dot{p}
		\end{bmatrix}= \begin{bmatrix}
			0_{n\times n}&I_n \\-I_n&-D(q,p)
		\end{bmatrix}\nabla H(q,p)+ \begin{bmatrix}
			0_{n\times m}\\G
		\end{bmatrix}u,\\
		&H(q,p)=\frac{1}{2}p^\top M^{-1}(q)p + V(q), \quad y= G^{\top}\dot{q},
	\end{split}
\end{equation} 
where $q,p \in \mathbb{R}^n$ are the generalized positions and momenta vectors, respectively, $u\in \mathbb{R}^m$ is the control vector, $y\in \mathbb{R}^m$ is the passive output, $m\leq n$, ${D(q,p): \mathbb{R}^n\times\mathbb{R}^n \to \mathbb{R}^{n \times n}}$ is the positive semi-definite damping matrix, $H: \mathbb{R}^n\times\mathbb{R}^n \to \mathbb{R}$ is the Hamiltonian, $M(q):\mathbb{R}^n \to \mathbb{R}^{n \times n} $ is the positive definite mass-inertia matrix, $V(q):\mathbb{R}^n\to \mathbb{R}$ is the potential energy of the system, and $G\in \mathbb{R}^{n\times m}$ is the constant input matrix, where $\mbox{rank}(G)=m$.

To formulate the problem under study, we identify the set of assignable equilibria for \eqref{phsys} given by
\begin{equation*}
	\mathcal{E} = \left\lbrace (q,p)\in\rea^{n}\times\rea^{n} \mid G^{\perp}\nabla V(q) = 0_{n}, p=0_{n}\right\rbrace
\end{equation*}
where $G^\perp$ is the (full rank) left annihilator of $G$. Then, we consider the PID-PBCs
\begin{equation}
	u=-K_Py-K_I(\gamma(q)+\kappa)-K_D\dot{y} \label{pid}
\end{equation} 
where the gains $K_P,K_I,K_D \in \mathbb{R}^{m\times m}$ satisfy $K_P,K_I>0$, and $K_D\geq0$, $\gamma(q):=G^{\top}q$, and $\kappa \in \mathbb{R}^m$ is a constant vector that is used to assign the equilibrium for the closed-loop system. Hence, the closed-loop system \eqref{phsys}-\eqref{pid} takes the form
\begin{equation}
	\begin{bmatrix}
		\dot{q}\\ \dot{p}
	\end{bmatrix}=\Upsilon^{-1}(q)F(q,p)\Upsilon^{-\top}(q)\nabla H_d(q,p) \label{phsyscl}
\end{equation} 
with \footnote{The closed-loop system \eqref{phsyscl} preserves the mechanical structure such as the IDA-PBC procedure (see \cite{GOMEZESTERN2004451}).}
\vspace{-2pt}
\begin{align}\label{clmatrices}
	H_d(q,p)&:=H(q,p)+\frac{1}{2}\lVert\gamma(q)+\kappa\rVert^2_{K_I}+\frac{1}{2}\lVert y \rVert^{2}_{K_{D}},\nonumber\\
	\Upsilon(q)&:=\begin{bmatrix}I_{n} & 0_{n\times n} \\ GK_D(\nabla_q y)^\top & I_{n}+GK_{D}G^{\top}M^{-1}(q)\end{bmatrix},\nonumber\\
	F(q,p)&:=\begin{bmatrix}0_{n\times n}&I_n\\-I_n & -D(q,p)-GK_PG^\top\end{bmatrix}.
\end{align}
Now, we formulate the problem under study as follows.

\textbf{Problem setting:} given $(q_{\star},0_{n})\in\mathcal{E}$, propose a method to choose the gains $K_P,K_I$, and $K_D$ of \eqref{pid} such that \eqref{phsyscl} has a \textit{local} transient response behavior that does not exhibit oscillations, and a prescribed damping ratio or rise time.

To develop the subsequent section, we introduce the following assumption:
\begin{assumption}\label{ass}
	The desired hamiltonian $H_d$ has a local isolated minimum at ${x_\star:=(q_\star,0)}$, i.e., ${x_\star=\arg\min H_d(q,p)}$.
\end{assumption}
\vspace{1mm}
\begin{remark}
	Assumption \ref{ass} implies that \eqref{phsyscl} has a stable equilibrium point at $x_\star$. Asymptotic stability follows if the output $y$ is detectable for the closed-loop system. Furthermore, \eqref{pid} defines an \textit{output strictly passive} operator $y \mapsto -u$. Hence, the application of the Passivity Theorem ensures that the closed-loop system is $\mathcal{L}_2$ stable (see \cite{van2000l2}). See \cite{zhang2017pid}, \cite{borja2020} for further details about the stability analysis of \eqref{phsyscl}. 
\end{remark}
\begin{remark}
	As indicated in \cite{borja2020}, the implementation of the term $K_{D}\dot{y}$ is subject to two conditions, namely, (i) $y$ must be of \textit{relative degree} one, and (ii) $u$ can be expressed as a function of the state vector without singularities. However, for systems of the form \eqref{phsys}, the mentioned conditions are always verified, where (i) is directly satisfied, and some simple computations show that the controller has no singularities since the matrix $\Psi(q):=I_{m} + K_{D}G^{\top}M^{-1}(q)G$ has full rank for every $K_{D}\geq 0$.
\end{remark}
\begin{remark}
The main difference between PID-PBCs and classical PID controllers is that the former are constructed around the \textit{passive output signal}, while the latter are designed in terms of an \textit{error signal}. In particular, for mechanical systems, the passive output signal is given by the actuated velocities, i.e., $G^\top\dot{q}$, while the error signal used to design classical PID controllers is given in terms of position. Some simple computations show that, for fully actuated mechanical systems, a PI-PBC scheme coincides with a classical PD controller.
\end{remark}

%===============================================================================
\vspace{-1mm}
\section{Tuning rules}\label{tuningrules}
In this section, we describe our approach to obtain the tuning rules. To facilitate the analysis, we linearize the system and convert the drift vector field into a class of saddle point matrices by similarity transformation. The main benefit of this particular form is that this reveals a clear relationship between the damping, the potential energy, and the kinetic energy, which is used later to propose the tuning rules. 
\subsection{Linearizing and obtaining the saddle point form}
To obtain the linearized dynamics of \eqref{phsyscl}, we introduce the vectors $\tilde{q}:=q-q_\star,\ \tilde{p}:=p.$

Consequently, the linearized system around the equilibrium point $(q_\star,0)$ corresponds to:
\begin{align}\label{phsyslin}
	\begin{bmatrix}
		\dot{\tilde{q}}\\ \dot{\tilde{p}}
	\end{bmatrix}=\Upsilon_\star^{-1}F_\star\Upsilon_\star^{-\top}\nabla^2 H_{d\star} 	\begin{bmatrix}
		\tilde{q}\\\tilde{p}
	\end{bmatrix}
\end{align}
where $\Upsilon$ and $F$ are defined as in \eqref{clmatrices}.
Then, to obtain the saddle point, we define 
\begin{align}\label{def1}
	\begin{split}
		\mathcal{R}:&=GK_PG^\top + D_\star \\
		\mathcal{P}:&=GK_IG^\top + \nabla^2 V_\star\\
		\mathcal{W}:&= GK_DG^\top + M_\star
	\end{split}
\end{align}
and $\phi_P,\phi_W\in \mathbb{R}^{n\times n}$ are full rank matrices satisfying \footnote{$\phi_P$ and $\phi_M$ are unique matrices obtained from the Cholesky decomposition.}
\begin{align}\label{phiM}
	\mathcal{W}^{-1}=\phi_W^\top \phi_W,\quad \mathcal{P}=\phi_P^\top \phi_P.
\end{align}
Subsequently, we define the similarity transformation matrix $T\in \mathbb{R}^{2n\times2n}$ and new coordinates $z\in \mathbb{R}^{2n}$ such that \footnote{The spectrum is invariant to similarity transformation.}
\begin{align}
	T:=\begin{bmatrix}
		0_{n\times n}&\phi_W^{-\top}M^{-1}_\star\\
		\phi_P&0_{n\times n}
	\end{bmatrix},\quad z:=T\begin{bmatrix}\tilde{q}\\\tilde{p}\end{bmatrix}.
\end{align}
Therefore, the linearized system in $z$ corresponds to
\begin{align}\label{syssaddle}
	\dot{z} = -\mathcal{N} z, \quad \mathcal{N}:=\begin{bmatrix}
		\phi_W \mathcal{R}\phi_W^\top&\phi_W\phi_P^\top\\
		-\phi_P\phi_W^\top & 0_{n\times n}
	\end{bmatrix}.
\end{align}
where $\mathcal{N}$ belongs to a class of saddle point matrices as \eqref{sys} with $X:=\phi_W \mathcal{R}\phi_W^\top$, $Z:=\phi_P\phi_W^\top$ and $Y:=0_{n\times n}$.

Finally, to characterize the eigenvalues of $\mathcal{N}$, denote with $(\lambda_{\mathcal{N}},v)$ an eigenpair of \eqref{syssaddle} with $\lambda_{\mathcal{N}} \in \mathbb{C}$ and $v \in \mathbb{C}^n$. Then, $\lambda_{\mathcal{N}}$ is given by the following expression (see \cite{benzi2006eigenvalues}):
\begin{align}\label{eig}
	\lambda_{\mathcal{N}}:=\frac{1}{2}\begin{pmatrix}
		\frac{v^*\phi_W \mathcal{R}\phi_W^\top v}{v^*v}\pm\sqrt{\Big(\frac{v^*\phi_W \mathcal{R}\phi_W^\top v}{v^*v}\Big)^2-4\frac{v^*\phi_W \mathcal{P}\phi_W^\top v}{v^*v}}
	\end{pmatrix}.
\end{align}

The terms $\mathcal{R},\mathcal{P}$, and $\mathcal{W}$ are associated with the damping injection, the potential energy, and the kinetic energy, respectively. In the sequel, we propose three conditions on the relation between these terms. When these conditions hold, the closed-loop system exhibits the specified --or desired-- transient response.  On the other hand, note that $\mathcal{R}$, $\mathcal{P}$, and $\mathcal{W}$ are related to the gain matrices $K_P$, $K_I$, and $K_D$ via \eqref{def1}. Hence, we can design the gain matrices such that the mentioned conditions are verified. 

\subsection{Removing the overshoot}
The oscillations of the transient response are characterized by the dominant pair of complex-conjugated poles of the system. The peak of such oscillations corresponds to the maximum overshoot of the system \cite{kulakowski2007dynamic}. Here, we provide a condition such that system \eqref{phsyslin} presents a “no-overshoot” response. In other words, the matrix $\mathcal{N}$ from system \eqref{phsyslin} must contain only real spectrum.

From Theorem \ref{t1}, an eigenvalue of $\mathcal{N}$ is real if and only if condition \eqref{condt1} holds, that is, the discriminant of \eqref{eig} is nonnegative. Then, to extend condition \eqref{condt1} to all the eigenvalues of $\mathcal{N}$, we propose the following:
\begin{proposition}\label{p1}
	The spectrum of the system \eqref{phsyslin} is real and nonnegative if the following is satisfied:
	\begin{align}\label{p1eq}
		\begin{split}
			4\lambda_{max}(\mathcal{P})\lambda_{max}(\mathcal{W})\leq\lambda_{min}(\mathcal{R})^2
		\end{split}
	\end{align}
\end{proposition}
\begin{proof}
	Let $\eta:=\phi_W^\top v$, then, expression \eqref{condt1} can be rewritten as
	\begin{align}
		4\frac{\eta^*\mathcal{P}\eta}{\eta^*\eta}\frac{\eta^*\mathcal{W}\eta}{\eta^*\eta}\leq\Bigg(\frac{\eta^*\mathcal{R} \eta}{\eta^*\eta}\Bigg)^2.
	\end{align}
	Therefore, if condition \eqref{p1eq} holds, then inequality \eqref{condt1} is satisfied for any $\lambda$ since
	\begin{align}
		\begin{split}
			4\frac{\eta^*\mathcal{P}\eta}{\eta^*\eta}\frac{\eta^*\mathcal{W}\eta}{\eta^*\eta}&\leq 4\lambda_{max}(\mathcal{P})\lambda_{max}(\mathcal{W})\\
			&\leq(\lambda_{min}(\mathcal{R}))^2\leq\Bigg(\frac{\eta^*\mathcal{R} \eta}{\eta^*\eta}\Bigg)^2.
		\end{split}
	\end{align}
	Any eigenvalue of $\mathcal{N}$ is characterized by expression \eqref{eig}, then it follows that
	\begin{align}
		\begin{split}
			0\leq\frac{\lambda_{\min}(\mathcal{R})}{\lambda_{\max}(\mathcal{W})}\leq\frac{v^*\phi_W \mathcal{R}\phi_W^\top v}{v^*v} \implies \Re(\lambda)\leq0.
		\end{split}
	\end{align}
\end{proof}
\begin{remark}
	The equality case in \eqref{p1eq} corresponds to a critical damped response. 
\end{remark}

\subsection{Prescribing the bounds for the damping ratio} 
The tuning rule provided with Proposition \ref{p1} might be restrictive for some applications that need a faster rise time. However, this is usually achieved at the expense of a transient response with overshoot and oscillations. If this performance is acceptable, we propose a rule to improve the rise time by tuning the bounds of the damping ratio of the spectrum of \eqref{phsyslin}. 

Denote with $\lambda_{\mathcal{N}} \in \mathbb{C}$  \textit{any} eigenvalue of $\mathcal{N}$, then, the standard definition of the damping ratio of $\lambda_{\mathcal{N}}$ is given by \cite{aastrom2010feedback}:
\begin{align}\label{dampratio}
	\zeta_\mathcal{N}:=\frac{|\Re(\lambda_{\mathcal{N}})|}{\sqrt{\Re(\lambda_{\mathcal{N}})^2+\Im(\lambda_{\mathcal{N}})^2}}
\end{align}
where $0\leq\zeta_\mathcal{N}\leq1$.

From \eqref{dampratio}, note that the damping ratio of the spectrum of $\mathcal{N}$ belongs to the interval $[0,1]$, which is conservative.  We can rewrite the definition of \eqref{dampratio} in terms of $\mathcal{R},\mathcal{P}$ and $\mathcal{W}$ and provide less conservative bounds, i.e.,

\begin{proposition}\label{p2}
	Denote with $(\lambda_{\mathcal{N}},v)$ \textit{any} eigenpair of \eqref{syssaddle} with $\lambda_{\mathcal{N}} \in \mathbb{C}$ and $v \in \mathbb{R}^n$, then, the damping ratio of $\lambda_{\mathcal{N}}$ is given by
	\begin{align}\label{zeta}
		\zeta_\mathcal{N} := \frac{1}{2}\frac{v^*Xv}{v^*v}\Bigg(\sqrt{\frac{v^*Z^\top Z v}{v^*v}}\Bigg)^{-1}
	\end{align}
	where this is bounded by
	\begin{align}\label{dampbound}
		\zeta_{\min} \leq \zeta_\mathcal{N}^2 \leq \zeta_{\max},
	\end{align}
	where
	\begin{align}
		\begin{split}
			\zeta_{\min}:=\max\Bigg\{0,\frac{1}{4}\frac{\lambda_{min}(\mathcal{R})^2}{\lambda_{max}(\mathcal{W})\lambda_{max}(\mathcal{P})}\Bigg\}\\
			\zeta_{\max}:=\min\Bigg\{1,\frac{1}{4}\frac{\lambda_{max}(\mathcal{R})^2}{\lambda_{min}(\mathcal{W})\lambda_{min}(\mathcal{P})} \Bigg\}.
		\end{split}
	\end{align}
\end{proposition}
\begin{proof}
	From the proof of Corollary \ref{t2} (see \cite{shen2010eigenvalue}), we have that:
	\begin{align}\label{thea}
		\Re(\lambda_{\mathcal{N}})=\frac{1}{2}\frac{v^*Xv}{v^*v}.
	\end{align}
	Note that expression \eqref{eig} follows from solving the quadratic equation (see \cite{benzi2006eigenvalues})
	\begin{align}\label{lambdaeq}
		\lambda_{\mathcal{N}}^2-\frac{ v^* X v}{v^* v}\lambda_{\mathcal{N}}+ \frac{v^* Z^\top Z v}{v^* v}=0.
	\end{align}
	Substituting \eqref{thea} in \eqref{lambdaeq} yields
	\begin{align}\label{theb}
		\frac{v^*Z^\top Z v}{v^*v}=\Re(\lambda_{\mathcal{N}})^2+\Im(\lambda_{\mathcal{N}})^2,
	\end{align}
	Then, expression \eqref{zeta} follows from substituting \eqref{thea} and \eqref{theb} in \eqref{dampratio}. By rewriting \eqref{zeta} we have that
	\begin{align}
		\zeta^2&=\frac{1}{4}\frac{(\eta^* \mathcal{R} \eta)^2}{(\eta^*\mathcal{W}\eta)( \eta^* \mathcal{P} \eta)},
	\end{align}
	therefore, \eqref{dampbound} follows from:
	\begin{align*}
		\begin{split}
			\frac{\lambda_{min}(\mathcal{R})^2}{\lambda_{max}(\mathcal{W})\lambda_{max}(\mathcal{P})}\leq\frac{(\eta^* \mathcal{R} \eta)^2}{(\eta^*\mathcal{W}\eta)( \eta^* \mathcal{P} \eta)}\leq \frac{\lambda_{max}(\mathcal{R})^2}{\lambda_{min}(\mathcal{W})\lambda_{min}(\mathcal{P})}.\\
		\end{split}
	\end{align*}
\end{proof} 
\subsection{Prescribing the upper bound for the rise time}
In this section, we proceed to characterize the upper bound of the rise time for system \eqref{phsyslin} based on the work of \cite{shen2010eigenvalue}. We define the rise time $t_r\in \mathbb{R}_+$ as the time taken by the system to reach $98\%$ of its steady state value. The rise time is influenced directly by the real part of the pole closest to the imaginary axis. Consider the following three scenarios:
\begin{enumerate}[label=\textbf{S\arabic*:},wide, labelwidth=!, labelindent=0pt]
	\item The spectrum of $\mathcal{N}$ is purely real.
	\item All the elements of the spectrum of $\mathcal{N}$ have an imaginary part different from zero.
	\item Some elements of the spectrum of $\mathcal{N}$ are purely real, and some elements are complex (imaginary part different from zero).
\end{enumerate}
Based on this premise, we define $t_{ru} \in \mathbb{R}_+ $ as the \textit{nominal rise time} of the system, then, we propose the following:
\begin{proposition}\label{p3}
	Denote with $\Re(\lambda_u)$ the lower bound for the real part of the spectrum of $\mathcal{N}$. Then, the rise time of the response of \eqref{syssaddle} is bounded from above by $t_{ru} \in \mathbb{R}_+$ where this is defined as
	\begin{align}\label{p3eq}
		\begin{split}
			t_{ru}:=\frac{4}{\Re(\lambda_u) }
		\end{split}
	\end{align}
	where $\Re(\lambda_u)$ is given by
	\begin{align}\label{limitp1}
		\Re(\lambda_u)=
		\begin{cases}
			\min\{\lambda_{\min}(\mathcal{W}^{-1}\mathcal{R}),\lambda_{\min}(\mathcal{R}^{-1}\mathcal{P})\} \ &\text{if S1}\\
			\frac{1}{2}\lambda_{\min}(\mathcal{W}^{-1}\mathcal{R}) \ &\text{if S2}\\
			\min\{\frac{1}{2}\lambda_{\min}(\mathcal{W}^{-1}\mathcal{R}),\lambda_{\min}(\mathcal{R}^{-1}\mathcal{P})\}, \ &\text{if S3.}\\
		\end{cases}
	\end{align}
\end{proposition}
\begin{proof}: The decay ratio of System \eqref{syssaddle} is bounded by $\Re(\lambda_u)$, therefore, expression 
	\begin{equation}\label{p3eqp}
		\exp^{-\Re(\lambda_u) t_{ru}}=0.0183 
	\end{equation}
	calculates the upper bound of the time where all the outputs of the systems have reached $98\%$ of the desired equilibrium point. Expression \eqref{p3eq} follows immediately by rearranging  \eqref{p3eqp}. Finally, \eqref{limitp1} follows immediately from substituting \eqref{def1} in Corollary 1. \\
\end{proof}

\begin{remark}
	S1 can be ensured by using Proposition \ref{p1} while S2 can be ensured with the condition $\frac{1}{4}\frac{\lambda_{max}(\mathcal{R})^2}{\lambda_{min}(\mathcal{W})\lambda_{min}(\mathcal{P})}<1$ from Proposition \ref{p2}.
\end{remark}
\begin{remark}
	Proposition \ref{p3} may be used as a tuning rule in combination with Propositions \ref{p1} or \ref{p2}. For example, note that the pair  $\{\lambda_{\min}{(\mathcal{R})},\lambda_{\max}{(\mathcal{P}})\}$ is used in Proposition \ref{p1} to ensure a ``no-overshoot" behavior, therefore, the pair $\{\lambda_{\max}{(\mathcal{R})},\lambda_{\min}{(\mathcal{P}})\}$ may be used to prescribe the upper bound of the rise time. 
\end{remark}

\begin{remark}
	For implementation purposes, the expression $\lambda_{\min}(\mathcal{\mathcal{X}}^{-1}\mathcal{Y})$ can be approximated with $\frac{\lambda_{\min}(\mathcal{Y})}{\lambda_{max}{(\mathcal{X}})}$, however, this might be conservative since $\frac{\lambda_{\min}(\mathcal{Y})}{\lambda_{max}{(\mathcal{X}})}\leq\lambda_{\min}(\mathcal{\mathcal{X}}^{-1}\mathcal{Y})$. 
\end{remark}
\subsection{Discussion}\label{discussion}
Some additional observations from this section are discussed below:
\begin{enumerate}[label=\roman*),wide, labelwidth=!, labelindent=0pt]%,align=left, leftmargin=*]
	\item \textbf{About the natural damping}: note that the tuning rules require some knowledge of the natural damping $D(q,p)$ of the system, which can be challenging in practice. Nevertheless, we stress the fact that the tuning rules will work even with a rough estimate as the closed-loop system will remain stable. 
	Some caveats of working with the rough estimate include changes in the oscillatory behavior and deviation of the bounds. Such variations may provide some intuition about the real bounds of the natural damping. For example, when applying Proposition \ref{p1} to achieve a critically damped response, if the system presents an over-damped (resp. under-damped) response, then the real damping value is larger (resp. smaller) than the nominal.
	\item \textbf{Improving the performance of a stable system}: if the open-loop system \eqref{phsys} is \textit{stable}, then, the controller \eqref{pid} can be used to improve its performance. 
	\item \textbf{Underactuated systems}: when $m<n$, the right hand side of \eqref{p1eq} reads as $\lambda_{min}(\mathcal{R})^2 = \lambda_{min}(D_\star)^2$ because $\lambda_{\min}(GK_PG^\top)=0$. Therefore, if $D(q,p)>0$ (i.e. $\lambda_{min}(D_\star)$ is strictly positive), then Proposition \ref{p1} can be applied to reduce the oscillations for an underactuated system.
	\item \textbf{Region of validity}: the proposed tuning rules are based on the linearization of the system. Accordingly, the performance analysis represented in this section is valid only in a neighborhood of $(q_\star, 0_n)$. A method to estimate such a region consists of using the tools designed to estimate the region of attraction of the equilibrium of stable linearized systems. For further details, see \cite{khalil2002nonlinear} or \cite{kloiber2012estimating}.
	\item \textbf{Comparison with the pole placement approach}: the pole placement method is a linear control technique. Therefore, to apply this technique to control a nonlinear system, it is necessary to linearize it. Moreover, the stability results are restricted to the region where the linearization is valid. On the other hand, in the PID-PBC approach, the stability results are based on (nonlinear) Lyapunov analysis, and the linearization is required exclusively to \textit{analyze the performance}--in terms of oscillations, rise time, or damping ratio--of the closed-loop system.
\end{enumerate}

%===============================================================================
\section{Experimental Results}\label{example}
In this section, the PID-PBC \cite{zhang2017pid} to stabilize a 2 DoF Rigid Planar Manipulator, as shown in Fig.~\ref{man} (see \cite{quanser} for the Quanser reference manual), at the desired equilibrium $(q_\star,0_2)$ with ${q_\star=col(0.6,0.8)}$, and to prescribe a desired behavior to the transient response. 
\begin{figure}[h]
	\centering
	\includegraphics[width=0.2\textwidth]{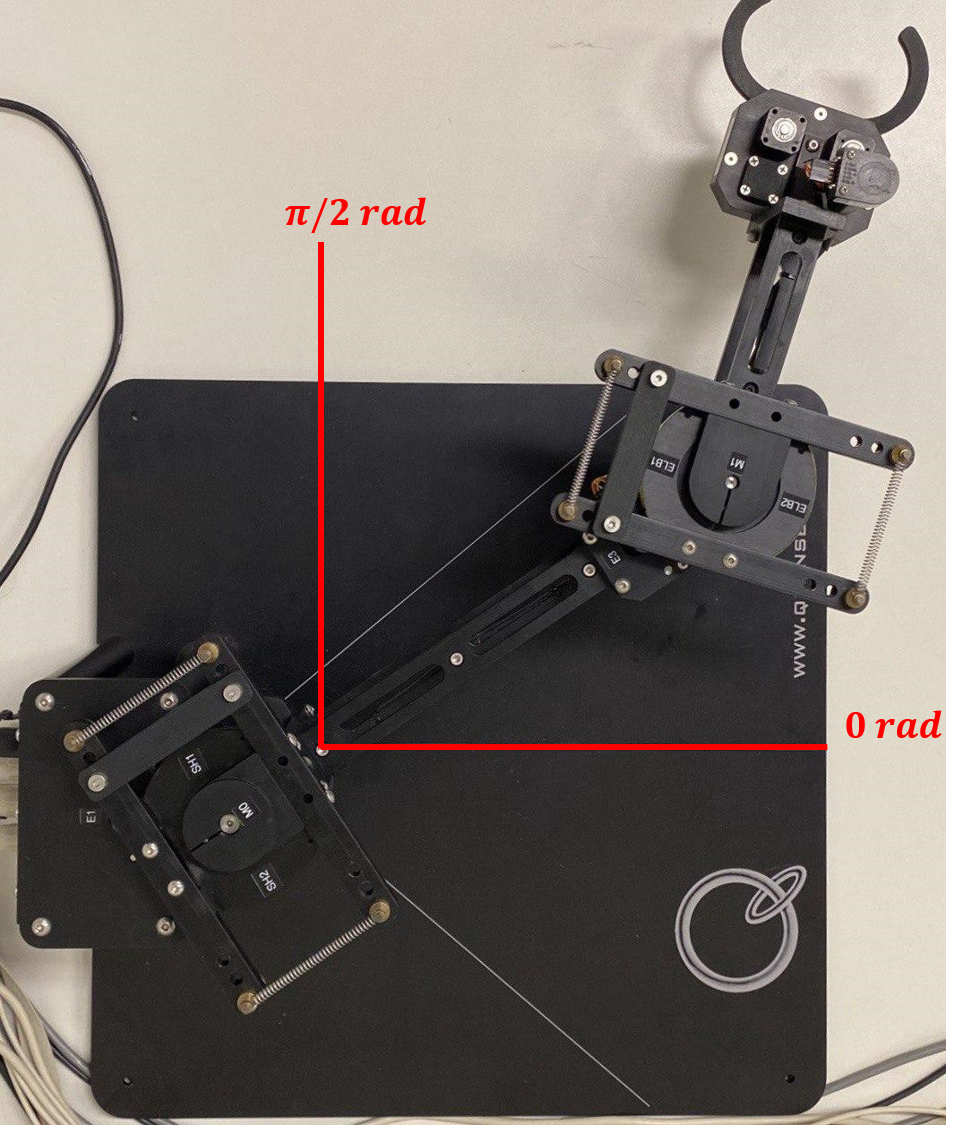}
	\caption{Experimental Setup: 2 DoF Planar Manipulator}\label{man}
	\vspace{-3mm}
\end{figure}

The manipulator model is described as in \eqref{phsys} with $n=2$, $V(q)=0$, $G=I_2$, $D= diag(0.07,0.03)$ and 
\begin{align}
	&M(q)=\begin{bmatrix}
		a_1+a_2+2b\cos(q_2)&a_2+b\cos(q_2)\\
		a_2+b\cos(q_2)&a_2
	\end{bmatrix},
\end{align}
where $a_1=0.1476$, $a_2=0.0725$, and $b=0.0858$.  

To illustrate the applicability of our tuning rules, we first obtain the results for a ``{Random~Tuning~(RT)}'' scenario for comparison purposes. Then, we perform the following experiments:
\begin{enumerate}[label=\textbf{E\arabic*:}]
	\item  system \textit{without} oscillations (Proposition \ref{p1}).
	\item  system \textit{with} oscillations (Proposition \ref{p2} with $0.4\leq~\zeta\leq~0.8$).
\end{enumerate} 
\begin{table}[t]
	\caption{Proportional, Integral and Derivative Gains}\label{gains}
	\centering
	\begin{tabular}{llll}
		\hline
		\multicolumn{1}{c}{} & RT        & E1              & E2              \\ \hline
		$K_P$                 & diag(1,0.5) & diag(7.3972,9.2) & diag(3.9136 4.1710) \\
		$K_I$                 & diag(50,30)     & diag(35,20)     & diag(50,45)     \\
		$K_D$                 & diag(0,0)  & diag(0,0)  & diag(0.08,0.15)  \\ \hline
	\end{tabular}
	\vspace{-2mm}
\end{table}
\begin{table}[t]
	\caption{Experimental rise time vs Nominal rise time}\label{trs}
	\centering
	\begin{tabular}{l|c|l|c|l|cl}
		\hline
		& \multicolumn{2}{c|}{RT}                         & \multicolumn{2}{c|}{E1}                              & \multicolumn{2}{c}{E2}                              \\ \cline{2-7} 
		& L1                         & \multicolumn{1}{c|}{L2} & L1                         & \multicolumn{1}{c|}{L2} & \multicolumn{1}{c|}{L1}    & \multicolumn{1}{c}{L2} \\ \hline
		Experimental (sec) & \multicolumn{1}{l|}{0.662} & 0.274                   & \multicolumn{1}{l|}{1.016} & 2.194                   & \multicolumn{1}{l|}{0.568} & 0.330                  \\ \hline
		Nominal (sec)  & \multicolumn{2}{c|}{3.397}                          & \multicolumn{2}{c|}{1.846}                          & \multicolumn{2}{c}{0.966}                          \\ \hline
	\end{tabular}
	\vspace{-3mm}
\end{table}
\begin{figure}[t]
	\centering
	\includegraphics[width=0.3\textwidth]{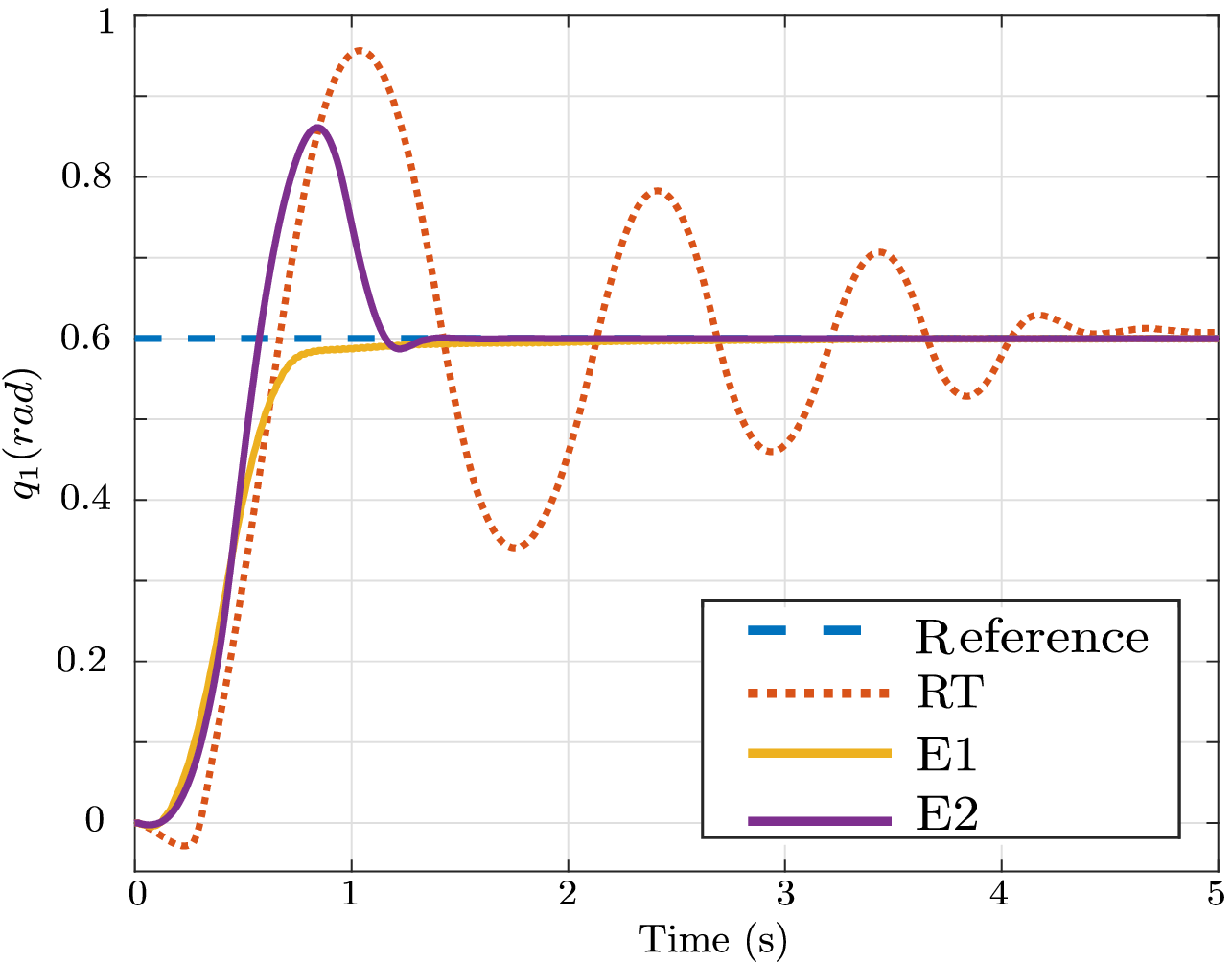}
		\vspace{-3mm}
	\caption{ Trajectories for angular position of L1. }\label{pos}
	\vspace{-3mm}
\end{figure}
\begin{figure}[t]
	\centering
	\includegraphics[width=0.3\textwidth]{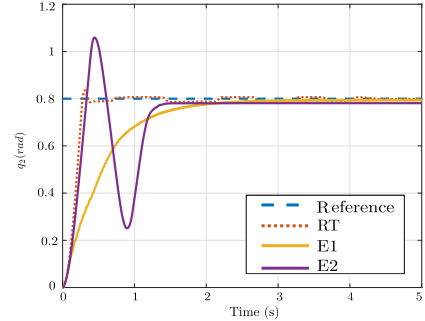}
		\vspace{-3mm}
	\caption{Trajectories for angular position of L2.}\label{pos2}
	\vspace{-6mm}
\end{figure}
Table \ref{gains} contains the gains calculated by using Propositions~\ref{p1} and \ref{p2} for each experiment. Furthermore, Table~\ref{trs} presents the upper bound estimation for the rise time using Proposition \ref{p3} for each experiment. The results of the angular position trajectories with initial conditions $(q,p)=0_4$ for Link 1 (L1) and Link 2 (L2) are shown in Fig.~\ref{pos} and Fig.~\ref{pos2}, respectively \footnote{The steady-state error presented in both figures is due to the nonlinearities no considered for our model, such as internal frictions and the dead zone of motors.}. For the reader convenience, Table~\ref{trs} shows the experimental results for the rise time from both experiments. Furthermore, a video with the experimental results can be found in this link: \url{https://youtu.be/aHPv-mKK_eI}.

Comparing E2 with E1 and RT, it can be seen, particularly in Link 2, that there is a trade-off between oscillations and the rise time, i.e., the faster the rise time, the more overshoot/oscillations the transient response exhibits. Additionally, note that tuning the kinetic term in E2 improves the settling time with respect to the RT scenario. 

Finally, although the nominal values in Table \ref{trs} are conservative, the rise time of each output is upper bounded; therefore, we can ensure that every trajectory has reached the $98\%$ of its final value by the nominal value. However, note that there is a particular case in E1 where the time taken for L2 is larger than the nominal. As mentioned in Section \ref{discussion}, as a consequence of working with a rough estimate of the natural damping, a deviation from the real value may occur. In this particular case, the nominal rise time is given by the expression \\[-2mm]
\begin{equation*}
	t_{ru}=4/\lambda_{\min}(\mathcal{R}^{-1}\mathcal{P})=4\lambda_{\max}(\mathcal{R}\mathcal{P}^{-1}),
\end{equation*}
where it can be seen that the rise time is proportional to the upper bound of the damping matrix $\mathcal{R}$.  Consequently, this rule suggests that the real damping is actually larger than the nominal provided by the manufacturer. 
%===============================================================================
\section{Concluding remarks and future work} \label{conclusion}
Our results have shown that transforming the pH structure into other coordinates reveals interesting spectral properties that can be used to improve the transient response for the nonlinear mechanical systems. Furthermore, it is clear from the tuning rules that there is an underlying relationship between the potential energy ($\mathcal{P}$), the kinetic energy ($\mathcal{W}$), and the damping ($\mathcal{R}$), which combinations result in a specific transient response. As seen in the experiments, the proposed tuning rules can prescribe the desired performance in terms of the oscillation, the damping ratio, and the rise time to a nonlinear MIMO mechanical systems.

As possible future research, we propose the implementation of damping identification methods in combination with PID-PBCs to apply the proposed tuning rules to a broader range of underactuated mechanical systems. Furthermore, we propose to establish tuning rules to prescribe a desired performance to the closed-loop system without linearizing it. For instance, tuning rules that prescribe an exponential rate of convergence for nonlinear system. Additionally, we aim to extend this methodology to other domains, such as electrical circuits or electro-mechanical systems. 
\bibliographystyle{ieeetr}
\bibliography{ref} 
\end{document}